\documentclass[12pt]{iopart}

\usepackage{graphicx,bbm,amssymb,amsopn}

\newcommand{\un}[1]{{\underline{#1}}}
\newcommand{\bra}[1]{{\langle #1 \vert}}

\newcommand{\ket}[1]{{\vert #1 \rangle}}
\newcommand{\xket}[1]{{\vert #1 \rangle}}
\newcommand{\braket}[2]{\langle #1 \vert #2 \rangle}
\newcommand{\xbraket}[2]{\langle #1 \vert #2 \rangle}

\newcommand{\ave}[1]{{\langle #1\rangle}}

\newcommand{\ii}{ {\rm i} }
\newcommand{\dd}{ {\rm d} }

\newcommand{\RR}{\mathbb{R}}
\newcommand{\CC}{\mathbb{C}}
\newcommand{\cc}{ {\hat c} }
\newcommand{\bb}{ {\hat b} }
\newcommand{\aaa}{ {\hat a} }

\newcommand{\LL}{{\hat {\cal L}}}
\newcommand{\MM}{{\hat {\cal M}}}

\newcommand{\NN}{{\hat {\cal N}}}

\newcommand{\mm}[1]{{\mathbf{#1}}}

\def\K{{\cal K}}

\def\tr{{\,{\rm tr}\,}}

\def\one{\mathbbm{1}}
\def\re{{\,{\rm Re}\,}}

\def\ness{\xket{{\rm NESS}}}

\newtheorem{theorem}{Theorem}[section]
\newtheorem{lemma}[theorem]{Lemma}

\newtheorem{corollary}[theorem]{Corollary}
\newtheorem{conjecture}{Conjecture}[section]

\newenvironment{proof}[1][Proof]{\begin{trivlist}
\item[\hskip \labelsep {\bfseries #1}]}{\end{trivlist}}

\newcommand{\qed}{\nobreak \ifvmode \relax \else
      \ifdim\lastskip<1.5em \hskip-\lastskip
      \hskip1.5em plus0em minus0.5em \fi \nobreak
      \vrule height0.75em width0.5em depth0.25em\fi}

\begin{document}

\title[Spectral theorem for Lindblad equation]{Spectral theorem for the Lindblad equation for quadratic open fermionic systems}
  
\author{Toma\v{z} Prosen}

\address{Department of physics, FMF, University of Ljubljana, Slovenia}
 

\date{\today}

\begin{abstract}
The spectral theorem is proven for the quantum dynamics of quadratic open systems of $n$ fermions described by the Lindblad equation.
Invariant eigenspaces of the many-body Liouvillean dynamics and their largest Jordan blocks are explicitly constructed for all eigenvalues.
For eigenvalue zero we describe an algebraic procedure for constructing (possibly higher dimensional) spaces of (degenerate) non-equilibrium steady states.
\end{abstract}

\maketitle

\section{Introduction}

We study spectral decomposition of a general quantum Lindblad equation \cite{lindblad,gorini,alicki} with quadratic generators \cite{fannes1,fannes2} for an open system of $n$ fermions. The physical picture and the motivation for the study of such a non-equilibrium quantum many-body problem have been outlined in Ref. \cite{njp} (see also \cite{prl}), on which we shall rely extensively.  However, in this note we show that an important assumption of diagonalizability made in Ref. \cite{njp} can be removed and that one can prove the spectral theorem for a general quadratic many-body Liouvillean and construct its canonical Jordan form.

The main results presented bellow, are: (i) computationally efficient\footnote{Computational problem of diagonalizing $4n\times 4n$ complex matrix of Ref.\cite{njp} has been reduced to a problem of diagonalizing a $2n\times 2n$ real matrix (or possibly finding its Jordan canonical form if it cannot be diagonalized) and/or
solving a $2n\times 2n$ linear matrix equation of Sylvester type \cite{bartels}.} canonical representation of the structure matrix of the Liouvillean quadratic form as summarized in lemmas \ref{lemma:zerodriving}, \ref{lemma:similar}, (ii) writing the general (possibly non-diagonalizable) {\em normal form} of the quadratic Lioivilleans in terms of normal master mode maps [formula (\ref{eq:dLL})],  (iii) constructing the full many-body Liouvillean spectrum, the invariant subspaces, and the largest Jordan block in each of them (theorem \ref{spctheorem}), and (iv) algebraic construction of the convex subspaces of possibly degenerate (multiple) non-equlibrium steady states (theorem \ref{uniqueness}).
 
A general master equation governing time evolution of the density matrix $\rho(t)$ of an open quantum system, preserving the trace and positivity of $\rho$, can be written in the Lindblad form
 \cite{lindblad} as  
\begin{equation}
\frac{\dd\rho }{\dd t} = \LL\rho :=
-\ii [H,\rho] + \sum_{\mu} \left(2 L_\mu \rho L_\mu^\dagger - \{L_\mu^\dagger L_\mu,\rho\} \right)
\label{eq:lind}
\end{equation}
where $H$ is a Hermitian operator (Hamiltonian), $[x,y]:=xy-yx$, $\{x,y\}:=xy+yx$, and $L_\mu$ are arbitrary operators, physically representing couplings to various baths.

We shall treat open quantum systems of a finite number, say $n$ fermions, with a {\em quadratic} generator $\LL$. Mathematically speaking, we shall restrict our discussion to the case where a set of $2n$ anti-commuting, Hermitian, bounded operators exist, say $w_j$, $j=1,\ldots,2n$,
\begin{equation}
\{w_j,w_k\} = 2\delta_{j,k}, \qquad j,k =1,\ldots, 2n,
\label{eq:CARw}
\end{equation}
so that the Hamiltonian, and the Lindblad operators, can be simultaneously expressed in terms of a {\em quadratic form}, and {\em linear forms}, respectively
\begin{eqnarray}
H &=& \sum_{j,k=1}^{2n} w_j H_{jk} w_k = \un{w} \cdot \mm{H}\, \un{w}, \label{eq:hamil}  \\
L_\mu &=& \sum_{j=1}^{2n} l_{\mu,j} w_j = \un{l}_\mu \cdot \un{w}.\  \label{eq:lindb}
\end{eqnarray}
Throughout this paper $\un{x}=(x_1,x_2,\ldots)^T$ will designate a vector (column)
of appropriate scalar valued or operator valued symbols $x_k$ and upright bold letter (e.g. $\mm{H}$) will designate
a complex or real matrix.
 
The main conceptual tool of our analysis is a Fock space ${\cal K}$ of operators (which may be referred to also as {\em Liouville-Fock} space) which the density operator $\rho(t)$ (\ref{eq:lind}) is a member of.
${\cal K}$ is in fact a $2^{2n} = 4^n$ dimensional linear vector space spanned by a canonical basis
 $\xket{P_{\un{\alpha}}}$ with
\begin{equation}
P_{\alpha_1,\alpha_2,\ldots,\alpha_{2n}} := 2^{-n/2} w_1^{\alpha_1}w_2^{\alpha_2}\cdots w_{2n}^{\alpha_{2n}}, \qquad
\alpha_j\in\{ 0,1\},
\end{equation}
which is  {\em orthonormal} with respect to Hilbert-Schmidt inner product
\begin{equation}
\xbraket{x}{y} = \tr x^\dagger y.
\label{eq:inner}
\end{equation}
Introducing  {\em creation and annihilation linear
maps} $\hat{c}_j$ over ${\cal K}$
\begin{equation}
\cc^\dagger_j \xket{P_{\un{\alpha}}} = \delta_{\alpha_j,0} \xket{w_j P_{\un{\alpha}}} \quad
\cc_j \xket{P_{\un{\alpha}}} = \delta_{\alpha_j,1} \xket{w_j P_{\un{\alpha}}}
\label{eq:defanih}
\end{equation}
the space ${\cal K}$ becomes a fermionic Fock space. {\em Canonical anticommutation relations} (CAR) among operator maps
$\{\cc_j,\cc_k\} = 0,\; \{\cc_j,\cc^\dagger_k\}=\delta_{j,k}$ and the fact that $\cc^\dagger_j$ is indeed a Hermitian adjoint of $\cc_j$ w.r.t. (\ref{eq:inner}) can be checked straightforwardly.

As shown in Ref.\cite{njp} the Lindblad equation (\ref{eq:lind}) in the Liouville-Fock picture
\begin{equation}
\frac{\dd}{\dd t}\ket{\rho} = \LL \ket{\rho}
\label{eq:liouville}
\end{equation}
takes manifestly quadratic form, where the Liouvillean map is expressed as \footnote{Note also a technical issue of a {\em parity symmetry} of
the Liouvillean and the operator space w.r.t. even and odd number of fermionic operator-excitations \cite{njp} which we may ignore here without affecting any of the
results.} 
\begin{equation}
\LL = \un{\aaa}\cdot\mm{A}\un{\aaa} - A_0 \hat{\one}
\label{eq:Amatrix}
\end{equation}
in terms of $4n$ Hermitian Majorana fermionic maps over the operator space $\K$, 
\begin{equation}
\aaa_{1,j} := (\cc_j + \cc_j^\dagger)/\sqrt{2},\quad
\aaa_{2,j} := \ii (\cc_j - \cc_j^\dagger)/\sqrt{2}
\end{equation}
satisfying CAR
\begin{equation}
\{\aaa_{\nu,j},\aaa_{\mu,k}\} = \delta_{\nu,\mu}\delta_{j,k},\quad \nu,\mu=1,2,\;\; j,k=1,\ldots,2n.
\end{equation}
We note that physically interesting {\em non-equilibrium steady state} (NESS), $\ness \in {\cal K}$, can be defined as the (possibly degenerate, and then non-unique)
right zero-eigenvalue eigenvector of the Liouvillean $\LL \ness = 0$.
The antisymmetric $4n \times 4n$ {\em structure matrix} $\mm{A}$ (\ref{eq:Amatrix}) has the general form \cite{njp}
\begin{eqnarray}
\mm{A} &=& 
\pmatrix{ 
-2\ii\mm{H} + 2\ii \mm{M}_{\rm i} & 2\ii \mm{M} \cr
-2\ii\mm{M}^T & -2\ii\mm{H} - 2\ii \mm{M}_{\rm i}}  \nonumber \\
&=& -2\ii \one_2 \otimes \mm{H} - 2 \sigma^2 \otimes \mm{M}_{\rm r} - 2(\sigma^1 - \ii \sigma^3)\otimes \mm{M}_{\rm i}
\label{eq:AA}
\end{eqnarray}
where $\sigma^\alpha,\alpha=1,2,3$ are the Pauli matrices, $\mm{H}=-\mm{H}^T$ is an antisymmetric (due to (\ref{eq:CARw})) $2n\times 2n$ Hermitean (purely imaginary) matrix, whereas 
\begin{equation}
\mm{M}_{\rm r} := \frac{1}{2}(\mm{M}+ \bar{\mm{M}}) = \mm{M}_{\rm r}^T,\qquad
\mm{M}_{\rm i} := \frac{1}{2\ii}(\mm{M} - \bar{\mm{M}}) = -\mm{M}_{\rm i}^T
\end{equation}
are the real and imaginary parts of the positive semidefinite Hermitian $2n\times 2n$ Lindblad {\em bath matrix} $\mm{M} := \sum_\mu \un{l}_\mu \otimes \bar{\un{l}}_\mu$, $\mm{M} \ge 0$. 
Note that, physically speaking, 
$\mm{H}$ encodes the Hamiltonian, $\mm{M}_{\rm r}$ encodes the {\em dissipation} and $\mm{M}_{\rm i}$ encodes the {\em driving}. 
The scalar $A_0$ (\ref{eq:Amatrix}) has a value
\begin{equation}
A_0 = 2\tr \mm{M} = 2\tr \mm{M}_{\rm r}.
\label{eq:A0}
\end{equation}
Writing the skew unit
\begin{equation}
\mm{J} := \sigma^1 \otimes \one_{2n}
\end{equation}
we observe that the structure matrix satisfies a particular self-conjugation property
\begin{equation}
\bar{\mm{A}} = \mm{J}  \mm{A} \mm{J}, \quad {\rm or\; equivalently} \quad \mm{A} \mm{J} = \mm{J} \bar{\mm{A}}.
\end{equation}

\section{Preliminaries}

A series of useful lemmas on the spectral decomposition of the structure matrix $\mm{A}$ shall now be straightforwardly proven, which explore
the structure of the problem and go beyond the general spectral theorems for the anti-symmetric complex matrices \cite{matrixtheory}.

\begin{lemma} Structure matrix (\ref{eq:AA}) is unitarily equivalent to a block-triangular matrix
\begin{equation}
\mm{\tilde A} = -2\ii\one_2\otimes \mm{H} - 2\sigma^3 \otimes \mm{M}_{\rm r} +
4\ii \sigma^+ \otimes \mm{M}_{\rm i} = \pmatrix{
-\mm{X}^T  & 4\ii \mm{M}_{\rm i}\cr
\mm{0} & \mm{X}} \label{eq:tA}
\end{equation}
where
\begin{equation}
\mm{X} := -2\ii \mm{H} + 2\mm{M}_{\rm r}
\label{eq:X}
\end{equation}
is a {\rm real} $2n \times 2n$ matrix.
\label{lemmatilde}
\end{lemma}
\begin{proof}
The similarity transformation $\mm{A}\to \tilde{\mm{A}}$
\begin{equation}
\tilde{\mm{A}} = \mm{U}\mm{A}\mm{U}^\dagger
\label{eq:Asim}
\end{equation}
is trivially provided by the cyclic permutation $\sigma^1 \to \sigma^2, \sigma^2 \to \sigma^3, \sigma^3 \to \sigma^1$, 
\begin{equation}
\mm{U} = \frac{1}{\sqrt{2}}\pmatrix{1 & -\ii \cr 1 & \ii} \otimes \one_{2n}.
\end{equation}
\qed
\end{proof} 

We shall refer to the representation in which the structure matrix has a form (\ref{eq:tA}) as the tilde-representation.
As an immediate consequence of lemma \ref{lemmatilde} we have:

\begin{corollary}
The characteristic polynomial of the structure matrix $\mm{A}$
factorizes in terms of the characteristic polynomial of $\mm{X}$ (\ref{eq:X})
\begin{equation}
p_\mm{A}(\beta) = p_\mm{X}(\beta) p_\mm{X}(-\beta),
\end{equation}
where $p_\mm{A}(\beta) := \det(\mm{A}-\beta \one_{4n}), p_\mm{X}(\beta) := \det(\mm{X}-\beta \one_{2n})$
and hence, if $\sigma(\mm{A})$ denotes the spectrum of a matrix,
\begin{equation}
\sigma(\mm{A}) = \sigma(\mm{X}) \cup (-\sigma(\mm{X})),
\end{equation}
where the union $\cup$ should keep track of (algebraic) multiplicities.
\end{corollary}
Curiously enough, this implies that the spectrum of the structure matrix is insensitive to the driving matrix $\mm{M}_{\rm i}$.
Note also that the real matrix $\mm{X}$ (\ref{eq:X}) has no general structure apart from the fact that $\mm{X} + \mm{X}^T = 4\mm{M}_r \ge 0$,
since $\mm{M}_r \ge 0$ is implied by $\mm{M} \ge 0$. Sometimes the eigenvalues $\{ \beta_j \}$ of $\mm{X}$ shall be referred to as {\em rapidities}.

\begin{lemma} 
\label{lemma:stability}
Let $\mm{X}$ be a  {\em real} square matrix, such that $\mm{X} + \mm{X}^T \ge 0$.
Then:
\begin{enumerate}
\item Any eigenvalue $\beta$ of $\mm{X}$ satisfies $\re \beta \ge 0$.
\item For any eigenvalue $\beta$ of $\mm{X}$ on the imaginary line, $\re \beta= 0$, its algebraic and geometric multiplicities coincide.
\end{enumerate}
\end{lemma}
\begin{proof}
(i) Take an eigenvalue $\beta$ and an eigenvector $\un{u}$, write $\mm{X}\un{u} = \beta \un{u}$, and the complex conjugate of this equation
$\mm{X}\bar{\un{u}} = \bar{\beta} \bar{\un{u}}$. Then take a dot product of the first equation with $\bar{\un{u}}$ and the dot product of the second
equation with $\un{u}$ and sum up:
\begin{equation}
\bar{\un{u}}\cdot (\mm{X}+\mm{X}^T)\un{u} = (2 \re \beta) \bar{\un{u}}\cdot \un{u}.
\end{equation}
Strict positivity of the eigenvector norm, $\bar{\un{u}} \cdot \un{u} > 0$, and non-negativity, $\bar{\un{u}}\cdot(\mm{X} + \mm{X}^T)\un{u} \ge 0$
(by definition of positive semidefiniteness of $\mm{X}+\mm{X}^T$), imply $\re \beta \ge 0$.

(ii) Consider a linear system of differential equations, 
\begin{equation}
(\dd/\dd t) \un{u}(t) = -\mm{X} \un{u}(t).
\label{eq:ode}
\end{equation} 
Positive semidefiniteness of $\mm{X} + \mm{X}^T$ is then equivalent to Lyapunov stability in control theory, namely $(\dd/\dd t)||\un{u}||^2_2 = -\bar{\un{u}}\cdot (\mm{X} + \mm{X}^T)\un{u} \le 0$ iff $\mm{X} + \mm{X}^T \ge 0$. Then, the converse (equivalent) version of the statement (ii) of the above lemma says that:
If there exists an imaginary (or vanishing) eigenvalue $\beta$, $\beta = \ii b$, and a corresponding Jordan block of dimension $k > 1$ in the Jordan canonical form of $\mm{X}$, then $\mm{X}+\mm{X}^T$ is not positive semidefinite.
This is indeed obvious, since if we take the initial vector $\un{u}(0)$ for (\ref{eq:ode}) from 
${\rm ker\,}(\mm{X} - \beta \one)^k \ominus {\rm ker\,}(\mm{X} - \beta \one)^{k-1}$ then 
$\un{u}(t) \propto t^{k-1} e^{-\ii b t}$, and so $\mm{X} + \mm{X}^T \not\ge 0$. 
 \qed
\end{proof}

\begin{lemma}
\label{lemma:zerodriving}
Let us write the Jordan canonical form of the matrix $\mm{X}$ (\ref{eq:X}) as
\begin{equation}
\mm{X} = \mm{P} \mm{\Delta} \mm{P}^{-1}
\label{eq:XJ}
\end{equation}
where $\mm{P}$ is a non-singular $2n\times 2n$ matrix (columns of which consist of generalized eigenvectors of $\mm{X}$) and
$\mm{\Delta} = \bigoplus_{j,k} \mm{\Delta}_{\ell_{j,k}}(\beta_j)$ is a direct sum of Jordan blocks
\begin{equation}
\Delta_\ell(\beta) := 
\pmatrix{ 
\beta & 1 & & \cr
  & \beta & \ddots & \cr
 & & \ddots & 1 \cr
 & & & \beta},
 \label{eq:canJB}
 \end{equation}
 for the distinct eigenvalues - rapidities $\beta_j$, and the block sizes $\ell_{j,k}$ satisfying the completeness sum
 \begin{equation} 
 \sum_{j,k} \ell_{j,k} = 2n.
 \label{eq:sumrule}
 \end{equation}
Then, the manifestly antisymmetric Jordan canonical form of the structure matrix (\ref{eq:AA}),
specialized to the case of zero driving $\mm{A}_0 := \mm{A}\vert_{\mm{M}_{\rm i} = 0}$,
\begin{equation}
\mm{A}_0 = 
\pmatrix{ 
-2\ii\mm{H}  & 2\ii \mm{M}_{\rm r} \cr
-2\ii\mm{M}_{\rm r} & -2\ii\mm{H}},
\label{eq:AA0}
\end{equation}
reads
\begin{equation}
\mm{A}_0 = 
\mm{V}_0^T \pmatrix{ \mm{0} & \mm{\Delta} \cr -\mm{\Delta}^T & \mm{0}} \mm{V}_0
\label{eq:A0J}
\end{equation}
where 
\begin{equation}
\mm{V}_0= (\mm{P}^T\oplus\mm{P}^{-1})\mm{U}
\label{eq:defV}
\end{equation} 
is a 
$4n\times 4n$ matrix (rows of which store the generalized eigenvectors of $\mm{A}_0$) which satisfies normalization
identity
\begin{equation}
\mm{V}_0\mm{V}_0^T = \mm{J}.
\label{eq:V0n}
\end{equation}
\end{lemma}     
\begin{proof}
Specializing lemma \ref{lemmatilde} to the case of $\mm{M}_{\rm i}=0$, we have 
$\tilde{\mm{A}}_0 = (-\mm{X}^T)\oplus \mm{X}$. Using the Jordan decomposition  (\ref{eq:XJ}) and the similarity transformation (\ref{eq:Asim}) we find
\begin{equation}
\mm{A}_0 = \mm{U}^\dagger (\mm{P}^{-T}\oplus \mm{P})[(-\mm{\Delta})^T\oplus \mm{\Delta}] (\mm{P}^T \oplus \mm{P}^{-1})\mm{U}.
\end{equation}
Introducing $\mm{V}_0:= (\mm{P}^T \oplus \mm{P}^{-1})\mm{U}$, i.e. equation (\ref{eq:defV}), and noting $\mm{U}^\dagger = \mm{U}^T \mm{J}$, we write 
\begin{equation} 
\mm{V}_0^{-1} = \mm{U}^\dagger (\mm{P}^{-T}\oplus \mm{P}) = 
\mm{U}^T (\mm{P}\oplus\mm{P}^{-T}) \mm{J} = \mm{V}^T_0 \mm{J}
\end{equation} 
which immediately implies (\ref{eq:A0J}) and (\ref{eq:V0n}).
\qed
\end{proof}

\begin{lemma}
\label{lemma:similar}
The general structure matrix $\mm{A}$ (\ref{eq:AA}) is always similar to its zero-driving counterpart $\mm{A}_0$ (\ref{eq:AA0})
\begin{equation}
\mm{A} = \mm{W}^{-1} \mm{A}_0 \mm{W}.
\label{eq:AWsim}
\end{equation}
The similarity transformation is provided by a complex orthogonal matrix 
\begin{equation}
\mm{W} = \one_{4n} +2 (\sigma^1 - \ii \sigma^3)\otimes\mm{Z},\qquad \mm{W} \mm{W}^T = \one_{4n},
\end{equation}
where $\mm{Z}$ is an antisymmetric $2n \times 2n$ real matrix, $\mm{Z}^T=-\mm{Z}$, obtained as a solution
of the continuous Lyapunov equation \cite{lyapunov}
\begin{equation}
\mm{X}^T \mm{Z} + \mm{Z} \mm{X} = \mm{M}_i
\label{eq:Lyap}
\end{equation}
which always exists (for the straightforward method of solution, see the proof).
\end{lemma}
\begin{proof}
Working in the tilde-representation (\ref{eq:Asim}), $\tilde{\mm{W}} := \mm{U} \mm{W} \mm{U}^\dagger = \one_{4n} -4\ii \sigma^+ \otimes \mm{Z}$, we have
$\tilde{\mm{W}}^{-1} =  \one_{4n} + 4\ii \sigma^+ \otimes \mm{Z}$ and
\begin{equation}
\tilde{\mm{W}}^{-1} \tilde{\mm{A}}_0 \tilde{\mm{W}} = \pmatrix{ -\mm{X}^T & 4\ii (\mm{X}^T\mm{Z}+ \mm{Z}\mm{X}) \cr \mm{0} & \mm{X}},
\end{equation}
so (\ref{eq:AWsim}) is indeed consistent with (\ref{eq:AA}) iff $\mm{Z}$ is antisymmetric and solves the continous Lyapunov equation (\ref{eq:Lyap}).
\footnote{All admisible structure matrices (\ref{eq:AA}) form a non-semisimple $n (6n-1)$ dimensional
{\em Lie algebra} spanned -- in a tilde-representation --
by $4n\times 4n$ matrices of the form
\begin{equation}
\pmatrix{-\mm{X}^T & \ii \mm{Y} \cr \mm{0} & \mm{X}}
\end{equation}
where $\mm{X}$ is an arbitrary $2n \times 2n$ real matrix and $\mm{Y}$ is an antisymmetric real matrix $\mm{Y}^T=-\mm{Y}$.
The similarity transformation (\ref{eq:AWsim}) is generated by a nilpotent subalgebra with $\mm{X}=0$.}

We shall now show by construction that under condition of positive semidefiniteness of $\mm{M} = \mm{M}_{\rm r} + \ii \mm{M}_{\rm i}$, the solution of (\ref{eq:Lyap}) always exists, but may not be unique if some of the eigenvalues $\beta_j$ of $\mm{X}$ (\ref{eq:X}) lie on the imaginary line.

Writing a trivial mapping ${\rm vec} : \RR^{2n \times 2n} \to \RR^{4n^2}$,  which vectorizes a matrix $({\rm vec\,} X)_{2n (i-1)+j} := X_{i,j}$,
and obeys ${\rm vec\,} (\mm{X} \mm{Z} \mm{Y}) = (\mm{X}\otimes \mm{Y}^T) {\rm vec\,}\mm{Z}$,
the continuous Lyapunov equation (\ref{eq:Lyap}) reads
\begin{equation}
(\mm{X}^T \otimes \one_{2n} + \one_{2n} \otimes \mm{X}^T)\,{\rm vec\,} \mm{Z} = {\rm vec\,} \mm{M}_{\rm i}.
\end{equation}
This equation is solved efficiently (in ${\cal O}(n^3)$ steps) using the Jordan canonical form of $\mm{X}$ (\ref{eq:XJ})
\begin{equation}
(\mm{\Delta}^T \otimes \one_{2n} + \one_{2n} \otimes \mm{\Delta}^T)\,{\rm vec\,}(\mm{P}^T \mm{Z} \mm{P}) = {\rm vec\,}(\mm{P}^T \mm{M}_{\rm i} \mm{P})
\label{eq:omegaeq}
\end{equation}
Note that the $(2n)^2 \times (2n)^2$ matrix 
$\mm{\Omega}:=\mm{\Delta}^T \otimes \one_{2n} + \one_{2n} \otimes \mm{\Delta}^T$ is {\em lower triangular} with diagonal entries
$\Omega_{r,r}$ of the form $\beta_j + \beta_{j'}$, so ${\rm vec\,}(\mm{P}^T \mm{Z} \mm{P})$ can be solved for uniquely if no such pairs $\beta_j$, $\beta_{j'}$ exist that $\beta_j+\beta_{j'}=0$ \cite{bartels,sylvester}.
Since $\re \beta_j \ge 0$ the problem ($\Omega_{r,r} = 0$) may arise in either of the two cases:
 (i) $\beta_j=0,\beta_{j'}=0$, or (ii) $\beta_j=\ii b,\beta_{j'}=-\ii b$, $b\in\RR\setminus\{0\}$.
It follows from lemma \ref{lemma:stability} that in either of these cases (i,ii) the corresponding Jordan blocks are trivial (of dimension 1), so there are no non-vanishing elements in the entire $r$-th row or column, $\Omega_{r,s}=\Omega_{s,r} = 0$.
Thus there will be a solution of (\ref{eq:omegaeq}) only if the corresponding coefficient on the RHS  vanishes 
\begin{equation}
\omega_r:=\left\{{\rm vec\,}(\mm{P}^T\mm{M}_{\rm i}\mm{P})\right\}_r=0,
\label{eq:omr}
\end{equation}
and then, the solution may be {\em non-unique} as the coefficient $\{ {\rm vec\,}(\mm{P}^T\mm{Z}\mm{P})\}_r$ is {\em arbitrary}
(unless it is fixed to zero by the antisymmetry of $\mm{Z}$ like in the case $j=j', \beta_j=0$).
Let us now study these two singular cases:
\begin{enumerate}
\item $\beta_j=0$ and $\beta_{j'}=0$. Suppose eigenvalue $0$ has multiplicity $m={\rm dim\, ker\,}\mm{X}$, and denote the corresponding basis vectors
$\mm{\Delta} \un{e}_l = 0$, as $\un{e}_l$, $l=1,\ldots,m$.
The RHS coefficient $\omega_r$ (\ref{eq:omr}) can be any of the matrix elements of an {\em antisymmetric} $m \times m$ matrix
$K_{l,l'} := (\un{e}_l\otimes \un{e}_{l'})\cdot {\rm vec\,}(\mm{P}^T\mm{M}_{\rm i}\mm{P}) = (\mm{P}\un{e}_{l}) \cdot \mm{M}_{\rm i}(\mm{P}\un{e}_{l'})=-K_{l',l}.$
Note that the corresponding eigenvectors of $\mm{X}$, $\mm{P}\un{e}_l$, are {\em real} (since $\mm{X}$ is a real matrix), and hence the matrix 
$\mm{K}$ is {\rm real} as well.
Positive semidefiniteness of $\mm{M}$ implies positive semidefiniteness of its counterpart reduced on the null space 
$Q_{l,l'} := (\mm{P}\un{e}_l) \cdot \mm{M} (\mm{P}\un{e}_{l'}) =  (\mm{P}\un{e}_l) \cdot \{\frac{1}{4}(\mm{X} + \mm{X}^T) +\ii \mm{M}_{\rm i}\} (\mm{P}\un{e}_{l'}) = \ii K_{l,l'}$, $\mm{Q} \ge 0$.
Since $\mm{K}^T=-\mm{K}$, this is only possible if $K_{l,l'} \equiv 0$.
 \item $\beta_j=\ii b,\beta_{j'}=-\ii b$, $b\in\RR\setminus\{0\}$. Suppose eigenvalue $\ii b$ has multiplicity $m={\rm dim\, ker\,}(\mm{X} \pm \ii b \one_{2n})$,
 and denote the corresponding basis vectors
$\mm{\Delta} \un{f}^{\pm}_l = \pm\ii b \un{f}^{\pm}_l$, as $\un{f}^{\pm}_l$, $l=1,\ldots,m$.
Since $\mm{X}$ is real, the corresponding eigenvectors $\un{u}^{\pm}_l:=\mm{P} \un{f}^{\pm}_l$,
$\mm{X} \un{u}^{\pm}_l = \pm\ii b \un{u}^{\pm}_l$, can be chosen to be conjugate paired
$\un{u}^{-}_l = \overline{\un{u}^{+}_l}$ (i.e. the corresponding columns of $\mm{P}$ can be chosen to be conjugate paired).
The RHS coefficient $\omega_r$ (\ref{eq:omr}) can be any of the matrix elements of an {\em antihermitian} $m \times m$ matrix
$K_{l,l'} := (\un{f}^+_l\otimes \un{f}^-_{l'})\cdot {\rm vec\,}(\mm{P}^T\mm{M}_{\rm i}\mm{P}) = \un{u}^+_l \cdot \mm{M}_{\rm i}\un{u}^-_{l'} = -\bar{K}_{l',l}$.
Positive semidefiniteness of $\mm{M}^+:=\mm{M}$ and $\mm{M}^-:=\bar{\mm{M}}$ 
implies positive semidefiniteness of two Hermitian $m \times m$ matrices
$Q^{\pm}_{l,l'} := \un{u}^+_l \cdot \mm{M}^\pm \un{u}^-_{l'} =  \un{u}^+_l \cdot \{\frac{1}{4}(\mm{X} + \mm{X}^T) \pm \ii \mm{M}_{\rm i}\} \un{u}^-_{l'} = 
\pm \ii K_{l,l'}$, $\mm{Q}^\pm \ge 0$. $\ii\mm{K}$ and $-\ii\mm{K}$ can be simultaneously positive semideifinite only if $K_{l,l'} \equiv 0$.
\qed
\end{enumerate}
\end{proof}

\section{Jordan canonical form on the tensor product spaces}

To best of the author's knowledge the theory of Jordan canonical form on tensor product spaces has not been very much developed. However, being able to manipulate Jordan decompositions on tensor product spaces seems crucial for discussing the general non-diagonalizable situation of quantum Liouvillean dynamics.
In this section we are demonstrating two simple results, namely one on Jordan decomposition of a sum of Jordan blocks on tensor product spaces, and the other one on Jordan decomposition of a many-body nilpotent map on the fermionic Fock space whose structure matrix is a simple Jordan block. Both results will be 
needed for the proof of the spectral theorem later in section \ref{theorems}.
 
\subsection{Jordan canonical form of the sum of Jordan blocks on a tensor product space}

Let us consider a tensor product space $\CC^k \otimes \CC^l$ and ask for the Jordan canonical form of the following $kl \times kl$ matrix
\begin{equation}
\mm{\Delta}_{k,l} = \mm{\Delta}_{k}(\alpha)\otimes \one_l + \one_k \otimes\mm{\Delta}_{l}(\beta)
\label{eq:Deltakl}
\end{equation}
where the Jordan block matrix is defined in (\ref{eq:canJB}).
\begin{lemma}
\label{2nilp}
The matrix (\ref{eq:Deltakl}) has the following Jordan canonical form
\begin{equation}
\mm{\Delta}_{k,l} \simeq \bigoplus_{r=1}^{\min\{k,l\}} \mm{\Delta}_{k+l-2r+1}(\alpha+\beta).
\label{eq:tensorsum}
\end{equation}
\end{lemma}
The proof will be constructive, so we shall also show how to compute the corresponding generalized eigenvectors.
\begin{proof}
Without loss of generality we can assume $k \ge l$. 
Let $\ket{i,j},i=1,\ldots,k,j=1,\ldots,l$ denote the basis of $\CC^k \otimes \CC^l$. Clearly, since $\mm{\Delta}_r(\gamma) = \one_r + \gamma\mm{\Delta}_r $
(we write $\mm{\Delta}_r := \mm{\Delta}_r(0)$), it is enough to show (\ref{eq:tensorsum}) for the nilpotent case $\alpha=\beta=0$, and then at the end 
add $(\alpha+\beta)\one_{kl}$. 
The crucial tool will be the Newton's binomial formula on the tensor product sum
\begin{equation}
\mm{\Delta}_{k,l}^p = \sum_{r=0}^p \left({p\atop r}\right) \mm{\Delta}_k^r \otimes \mm{\Delta}_l^{p-r}.
\end{equation}
From the above formula, and nilpotency $\mm{\Delta}^r_r = \mm{0}$, $\mm{\Delta}^{r-1}_r \neq \mm{0}$,
it follows that $\mm{\Delta}_{k,l}^{k+l-1} = \mm{0}$ but $\mm{\Delta}_{k,l}^{k+l-2}= \left({k+l-2\atop l-1}\right) \mm{\Delta}^{k-1}_k\otimes\mm{\Delta}^{l-1}_l \neq \mm{0}$,
so the {\em largest} Jordan block in the decomposition of $\mm{\Delta}_{k,l}$ has size $k+l-1$, with the first nilpotent series of 
generalized eigenvectors
\begin{equation}
\ket{p;1} := \mm{\Delta}^p_{k,l} \ket{0;1}, \qquad p = 0,\ldots, k+l-2,
\label{eq:maximal}
\end{equation}
and a `seed' vector (null vector of $\mm{\Delta}_{k,l}^{k+l-2}$), $\ket{0;1}:=\ket{k,l}$.

Let us introduce the diagonal number matrix $\mm{N} := \sum_{i=1}^k\sum_{j=1}^l (i+j)\ket{i,j}\bra{i,j}$ which slices the vector space to a direct sum of eigenspaces of the following dimensions
\begin{equation}
\dim\ker (\mm{N}-(q+1)\one_{kl}) = \dim\ker(\mm{N}-(k+l+1-q)\one_{kl}) = q,
\end{equation}
which change exactly by 1 when changing the eigenvalue by 1.
Further, note that $ \mm{N}\mm{\Delta}_{k,l}= \mm{\Delta}_{k,l} (\mm{N}-\one_{kl})$, or equivalently $\mm{\Delta}_{k,l}$ always decreases the eigenvalue of $\mm{N}$ exactly by $1$
\begin{equation}
\mm{\Delta}_{k,l} \ker (\mm{N}- (q+1)\one_{kl}) \subseteq  \ker (\mm{N}- q\one_{kl}).
\end{equation}

Then we make the following inductive argument. Suppose that for all $q < r$, the series of generalized eigenvectors 
\begin{equation}
\ket{p;q}:=\mm{\Delta}^p_{k,l}\ket{0;q},\qquad p = 0,\ldots, k+l-2q,
\end{equation}
explore all the eigenspaces of $\mm{N}$ with eigenvalues $\le r$ and $\ge k+l+2-r$. This is trivially true for $r=2$, or the maximal series with $q=1$ which we have
constructed above (\ref{eq:maximal}).

We need to show now that the $r$-th series explores the other two extremal eigenvalues, $r+1$ and $k+l+1-r$ and is of length $k+l-2r+1$.
Clearly, the seed vector for the $r-$th series $\ket{0;r}$ should have the number eigenvalue $k+l+1-r$, i.e. it  should be of the form
\begin{equation}
\ket{0;r} = \sum_{q=1}^r c^{(r)}_q \ket{k-q+1,l-r+q},
\end{equation}
and the size of the corresponding Jordan block can not be longer than $k+l-2r+1$
\begin{equation}
\mm{\Delta}_{k,l}^{k+l-2r+1}\ket{0;r} = 
\sum_{j=1}^{r-1} \ket{r-j,j} \sum_{q=1}^r  \left({ k+l-2r+1\atop l-r-j+q}\right) c^{(r)}_q = 0.
\end{equation}
This results, for the $r$th Jordan block, in a homogeneous system of $r-1$ equations for $r$ unknown coefficients $c^{(r)}_j$,
\begin{equation}
\sum_{q=1}^r  \left({ k'+l'+1\atop l'-j+q}\right) c^{(r)}_q = 0,\quad j \in \{1,\ldots, r-1\},
\label{eq:iden}
\end{equation} 
where on RHS and below we write $k':=k-r, l':=l-r$.
The above system has (up to a prefactor) a unique solution
\begin{equation}
c^{(r)}_q = (-1)^{r-q}  \left({ k'+r-q\atop r-q}\right)  \left({l'+q-1\atop q-1}\right).
\label{eq:crq}
\end{equation}
What is left to show is that a nilpotent series starting from a seed $\ket{0;r}$ is indeed not shorter than $k+l-2r+1$, i.e.
$\mm{\Delta}_{k,l}^{k+l-2r} \ket{0;r} \neq 0$, or equivalently 
\begin{equation}
\sum_{q=1}^r  \left({ k'+l'\atop l'-j+q}\right) c^{(r)}_q \ne 0,
\label{eq:noniden}
\end{equation}
for at least some $j \in\{1,\ldots,r\}$.
Verification of (\ref{eq:iden}) and (\ref{eq:noniden}) with ansatz (\ref{eq:crq}) is an easy case for a computer-assisted proving system based on Zeilberger algorithm \cite{paule}. \qed
\end{proof}

\subsection{Jordan canonical form of the many-body nilpotent map}

\label{mbnilp}

Here we study the following abstract  many body problem. Let $\bb_{k},\bb'_{k}, k=1,\ldots,\ell $ be fermionic maps satisfying CAR,
$\{\bb_j,\bb_k\}=\{\bb'_j,\bb'_k\}=0,\{\bb_j,\bb'_k\}=\delta_{j,k}$, and
defining $\ell $ fermionic modes, i.e. spanning a $2^\ell$ dimensional vector space ${\cal V}$ with the Fock basis and its dual defined, respectively, as
\begin{eqnarray}
\ket{\nu_1,\ldots,\nu_{\ell}} :=  {(\bb'_1)}^{\nu_1} \cdots {(\bb'_{\ell})}^{\nu_{\ell}}\ket{{\rm right}},\\
\bra{\nu_1,\ldots,\nu_{\ell}} := \bra{{\rm left}}(\bb_{\ell})^{\nu_{\ell}} \cdots (\bb_1)^{\nu_1},\qquad \nu_k\in\{0,1\}. 
\end{eqnarray}
The vector $\ket{\rm right}$ and its dual $\bra{\rm left}$, $\braket{\rm left}{\rm right}=1$, are defined via $\bra{\rm left}\bb'_k \equiv 0, \bb_k \ket{\rm right}\equiv 0$, hence the biorthogonality $\braket{\un{\nu}}{\un{\nu}'} = \delta_{\un{\nu},\un{\nu}'}$ simply follows from the CAR.

We shall study the following {\em nilpotent} linear operator over ${\cal V}$
\begin{equation}
\MM = \sum_{k=1}^{\ell-1} \bb'_{k+1} \bb_k.
\label{eq:nilpotent}
\end{equation}
We observe immediately that $\MM$ commutes with a {\em number map} $\NN := \sum_{k=1}^\ell \bb'_k \bb_k$, $[\MM,\NN]=0$, which foliates the linear space to a direct sum ${\cal V} = \bigoplus_{m=0}^\ell {\cal V}_m$, of ${\rm dim\, }{\cal V}_m = \left({\ell\atop m}\right) = \frac{\ell!}{m!(\ell-m)!}$ dimensional spaces ${\cal V}_n$, which are  spanned by Fock vectors $\ket{\un{\nu}}$ with exactly $|\un{\nu}|:=\sum_{k=1}^\ell \nu_k = m$ `particle' excitations ($m$ being the eigenvalue of $\NN$).

We shall now {\em fix} the number of particles $m$ and restrict the nilpotent operator $\MM$ (\ref{eq:nilpotent}) to the space ${\cal V}_m$, $\MM_m := \MM|_{{\cal V}_m}$. For $m=0$ (and $m=\ell$) $\MM$ acts trivially, $\MM_0 = \MM_m = 0$, and for $m=1$ (and $m=\ell-1$), 
$\MM_1 \simeq \MM_{\ell-1} \simeq \mm{\Delta}_\ell$. However, the interesting question now is to find a Jordan canonical form of ${\cal M}_m$, 
for arbitrary $m$. 

Let us define a {\em weight} of a basis Fock vector $\ket{\un{\nu}}$ as
\begin{equation}
\omega_{\un{\nu}} = \sum_{k=1}^\ell k \nu_k - \frac{1}{2}m(m+1).
\end{equation}
Having $\ell,m$ fixed, there are two extremal states, with minimum and maximum weight, namely
\begin{eqnarray}
\ket{\rm min} &=& \ket{1,\ldots,1,0,\ldots,0},\quad \omega_{\rm min} = 0, \nonumber \\
\ket{\rm max} &=& \ket{0,\ldots,0,1,\ldots,1},\quad \omega_{\rm max} = (\ell-m)m.
\end{eqnarray}
Denoting the linear spans ${\cal V}_m^r = L\{ \ket{\un{\nu}}; |\un{\nu}|=m, \omega_{\un{\nu}} = r\}$, we have a direct sum
\begin{equation}
{\cal V}_m = \bigoplus_{r=0}^{(\ell-m)m} {\cal V}^r_m.
\end{equation}
Dimensionality ${\rm dim}{\cal V}^r_m =: \left({\ell\atop m}\right)_{r}$ has a clear combinatorial-physical meaning, namely $\left({\ell\atop m}\right)_r$ is the number of distributions (combinations) of $m$ particles on $\ell$ exclusive equidistant vertical levels, such that the total potential energy is exactly $r$ 
steps larger than the minimal. Such a {\em restricted binomial symbol} can be calculated from a straightforward recursion relation
\begin{equation}
\left({\ell\atop m}\right)_r =  \left({\ell-1\atop m}\right)_r + \left({\ell-1\atop m-1}\right)_{r-\ell+m}, \qquad
\left({0\atop m}\right)_r = \delta_{m,0} \delta_{r,0}
\end{equation}
and clearly satisfies the sum rule
\begin{equation}
\sum_{r=0}^{(\ell-m)m}  \left({\ell\atop m}\right)_r =  \left({\ell\atop m}\right)
\label{eq:sumrule1}
\end{equation}
and the symmetry
\begin{equation}
\left({\ell\atop m}\right)_r = \left({\ell\atop m}\right)_{(\ell-m)m-r}.
\end{equation}
Note that $ \left({\ell\atop m}\right)_r = 0$ if $r < 0$ or $r > m(\ell-m)$.
Crucial observation for the rest of our discussion is the fact that $\MM_m$ maps ${\cal V}_m^r$ to ${\cal V}_m^{r+1}$,
i.e. that the map $\MM$ (\ref{eq:nilpotent}) increases the weight of the state exactly by $1$,
\begin{equation}
\MM {\cal V}_m^r \subseteq {\cal V}_m^{r+1}.
\label{eq:incl}
\end{equation}

Clearly, one can now state the following:
\begin{lemma}
\label{nilplemma}
The largest Jordan block in the canonical form of $\MM_m$ has size $(\ell-m)m + 1$.
\end{lemma}
\begin{proof}
Taking $\ket{0;1}:=\ket{\rm min}$ as a `seed' for the first nilpotent series, we see that, due to (\ref{eq:incl}), 
\begin{equation}
\ket{p;1} := \MM^p \ket{\rm min} \in {\cal V}_{m}^{p}.
\end{equation}
All these vectors for $p \le (m-\ell)m$ should be non-vanishing since only positive integer coefficients can appear in expressing the series 
$\MM^p \ket{\rm min}$ in the canonical basis $\ket{\un{\nu}}$, and the only canonical basis vector $\ket{\nu}$ of ${\cal V}_m$ which is annihilated by 
$\MM$ is $\ket{\rm max}$. So we have $\ket{(\ell - m)m;1} = q\ket{\rm max}$, where $q$ is a positive integer since $\dim {\cal V}_m^{(\ell - m)m} = 1$, and 
$\MM\ket{(\ell-m)m;1} = 0$. \qed
\end{proof}

The rest of Jordan decomposition of $\MM_n$ we are only able to state if the following conjecture is true:
\begin{conjecture}
\label{conjecture}
The map $\MM : {\cal V}^r_m \to {\cal V}^{r+1}_m$ 
is {\em injective} for $r \le \lfloor ((\ell-m)m-1)/2\rfloor $ and {\em surjective} for
$r \ge \lfloor (\ell-m)m/2\rfloor.$
\end{conjecture}
Conjecture has been verified explicitly using computer algebra code for all $\ell \le 12$, and for any $\ell$ and $m \le 3$ \cite{semrl}, however its general proof remains open due to
a rather involved combinatorics.
The conjecture of course also implies the ordering on the restricted binomial sequence (dimensions of ${\cal V}^r_m$), namely
\begin{equation}
\left({\ell\atop m}\right)_0 \le \left({\ell\atop m}\right)_1 \le \ldots \le \left({\ell\atop m}\right)_{\lfloor (\ell-m)m/2 \rfloor}
\end{equation}
which can be verified directly.

Omitting a straightforward proof \cite{semrl} we state that the conjecture implies the following simple general result
\begin{equation}
\MM_n \simeq \bigoplus_{r=0}^{\lfloor (\ell-m)m/2\rfloor} \left[ \left({\ell\atop m}\right)_r - \left({\ell\atop m}\right)_{r-1}\right] \mm{\Delta}_{(\ell - m)m+1-2r}
\end{equation}
where the notation $[k]\mm{\Delta}$ means $\mm{\Delta}\oplus\mm{\Delta}\ldots\oplus\mm{\Delta}(k {\rm\; times})$, and $[0]\mm{\Delta} := 0$.
In other words, the number of Jordan blocks of size $(\ell-m)m+1-2r$ (it can only decrease from the maximal in steps of $2$) is
$\left({\ell\atop m}\right)_r - \left({\ell\atop m}\right)_{r-1}$ (which can also be zero for some instances $r$).
The total number of Jordan blocks and thus the total number of proper eigenvectors of $\MM_m$ is then simply equal to 
$\left({\ell\atop m}\right)_{\lfloor (\ell-m)m/2\rfloor}$.

\section{The spectral theorem}

\label{theorems}

With the lemmas \ref{lemma:zerodriving},\ref{lemma:similar} we have shown that any Liouvillean structure matrix of the form (\ref{eq:AA}) admits a canonical 
representation in terms of the Jordan canonical form (\ref{eq:XJ},\ref{eq:canJB}) of the real matrix $\mm{X}$ (\ref{eq:X}), and the anti-symmetric real matrix $\mm{Z}=-\mm{Z}^T$ which solves the Lyapunov equation (\ref{eq:Lyap}), namely
\begin{equation}
\mm{A} = 
\mm{V}^T \pmatrix{ \mm{0} & \mm{\Delta} \cr -\mm{\Delta}^T & \mm{0}} \mm{V}.
\label{eq:AJ}
\end{equation}
The eigenvector matrix (rows of which store the generalized eigenvectors of $\mm{A}$) reads
\begin{eqnarray}
\mm{V} &=& \mm{V}_0 \mm{W} = (\mm{P}^T \oplus \mm{P}^{-1} )\mm{U} 
(\one_{4n} +2 (\sigma^1 - \ii \sigma^3)\otimes\mm{Z}) \nonumber \\
&=& \frac{1}{\sqrt{2}}\pmatrix{\mm{P}^T(\one_{2n}-4\ii\mm{Z})  &  -\ii \mm{P}^T (\one_{2n}-4\ii\mm{Z}) \cr \mm{P}^{-1} & \ii \mm{P}^{-1}} 
\label{eq:V}
\end{eqnarray}
and satisfies the canonical normalization condition
\begin{equation}
\mm{V}\mm{V}^T = \mm{J}.
\label{eq:VJ}
\end{equation}

Let us name the first $2n$ rows of $\mm{V}$ as $\un{v}_{j,k,l}$, and the last $2n$ rows as $\un{v}'_{j,k,l}$, which are exactly the generalized eigenvectors pertaining to $k$-th Jordan block of the eigenvalue (rapidity) 
$\beta_j$, and $-\beta_j$ respectively, and $l=1,\ldots,\ell_{j,k}$ ($l=1$ designates the {\em proper} eigenvector) where $\ell_{j,k}$ is the size of the Jordan block $(j,k)$. 
Then we introduce the {\em normal master mode} (NMM) maps as
\begin{equation}
\bb_{j,k,l} := \un{v}_{j,k,l} \cdot \un{\aaa},\qquad \bb'_{j,k,l} := \un{v}'_{j,k,l}\cdot \un{\aaa},
\label{eq:bfroma}
\end{equation}
satisfying the almost-CAR (simply following from (\ref{eq:VJ}))
\begin{equation}
\{ \bb_{j,k,l}, \bb_{j',k',l'} \} = 0,\;\; \{ \bb_{j,k,l}, \bb'_{j',k',l'} \} = \delta_{j,j'} \delta_{k,k'} \delta_{l,l'},\;\; \{ \bb'_{j,k,l}, \bb'_{j',k',l'} \} = 0.
\label{eq:almostCAR}
\end{equation}
In terms of decomposition (\ref{eq:AJ}) and NMM maps (\ref{eq:bfroma}) the Liouvillean (\ref{eq:Amatrix}) gets almost-diagonal form
\begin{eqnarray}
\LL &=& (\mm{V}\un{\aaa}) \cdot \left[\mm{\Delta} \oplus (-\mm{\Delta}^{T})\right]\mm{J} (\mm{V}\un{\aaa}) - A_0 \hat{\one} \\
&=& \sum_{j,k} \Bigl\{ \beta_j \sum_{l=1}^{\ell_{j,k}} (\bb_{j,k,l} \bb'_{j,k,l} - \bb'_{j,k,l} \bb_{j,k,l}) - \ell_{j,k} \beta_j \hat{\one} \nonumber \\
&& \quad\;\;\; + \sum_{l=1}^{\ell_{j,k}-1} (\bb_{j,k,l} \bb'_{j,k,l+1} - \bb'_{j,k,l+1} \bb_{j,k,l}) \Bigr\}
\end{eqnarray}
where the scalar (\ref{eq:A0}) has been expressed as $A_0 = \tr \mm{X} = \sum_{j,k} \ell_{j,k} \beta_j$.
After exercising (\ref{eq:almostCAR}) the Liouvillean finally simplifies to a {\em normal form}
\begin{equation}
\LL = -2\sum_{j,k} \left\{ \beta_j \sum_{l=1}^{\ell_{j,k}}  \bb'_{j,k,l} \bb_{j,k,l} + \sum_{l=1}^{\ell_{j,k}-1} \bb'_{j,k,l+1} \bb_{j,k,l} \right\}.
\label{eq:dLL}
\end{equation}
Note that the second (non-diagonal) sum vanishes for the trivial Jordan blocks with $\ell_{j,k}=1$.

Following Ref. \cite{njp} a {\em unique} pair\footnote{Unique with respect to a chosen, 
possibly non-unique (see lemma \ref{lemma:similar}) decomposition of the structure matrix (\ref{eq:AJ}).} of operator Fock states exist, satisfying 
\begin{equation}
\bb_{j,k,l} \ness = 0,\qquad \bra{1} \bb'_{j,k,l} = 0.
\label{eq:vacua}
\end{equation}
The bra (dual) vector $\bra{1}$ in fact represents a trivial operator $1$, dual to $\ness$, $\braket{1}{\rm NESS}={\rm tr}\rho_{\rm NESS} = 1$. Indeed $\ness$, which is an empty Fock state w.r.t. creation NMM maps $\bb'_{j,k,l}$, is a {\em non-equilibrium steady state} (NESS) since straightforward inspection following from almost-diagonal form of $\LL$ (\ref{eq:dLL}) shows that $\LL \ness = 0$, as well as the dual relation $\bra{1}\LL = 0$.

The complete spectrum of eigenvalues and eigenvectors of $4^n \times 4^n$ matrix of $\LL$ over ${\cal K}$ is trivial to construct in case that all Jordan blocks of $\mm{X}$ are trivial ($\ell_{j,k}=1$), i.e. 
if $\mm{A}$ is diagonalizable, the case which has been assumed in Ref.\cite{njp}. However, here we treat the general case.
 
\subsection{Spectral decomposition and Jordan canonical form of the Liouvillean}

\begin{theorem}
\label{spctheorem}
(i) The complete spectrum of Liouvillean $\LL$ is given by the following integer linear combinations
\begin{equation}
\lambda_{\un{m}} = -2\sum_{j,k} m_{j,k}\beta_j\,,\quad m_{j,k} \in \{0,1,\ldots,\ell_{j,k}\}.
\end{equation}
(ii) The $4^n$ dimensional operator space, and its dual (the bra-space), admit the following decomposition
\begin{equation}
{\cal K} = \bigoplus_{\un{m}} {\cal K}_{\un{m}}\, , \qquad {\cal K}' = \bigoplus_{\un{m}} {\cal K}'_{\un{m}}
\label{eq:direct}
\end{equation}
in terms of $\dim {\cal K}_\un{m} = \prod_{j,k}\left({\ell_{j,k}\atop m_{j,k}}\right)$ dimensional 
invariant subspaces 
\begin{equation}
\LL{\cal K}_{\un{m}} \subseteq {\cal K}_{\un{m}}\, ,\qquad{\cal K}'_{\un{m}}\LL \subseteq {\cal K}'_{\un{m}}
\label{eq:inv}
\end{equation}
spanned by
\begin{eqnarray}
{\cal K}_{\un{m}} = L\left\{ \prod_{j,k}\prod_{\eta=1}^{m_{j,k}} \bb'_{j,k,l_\eta}  \ness;\; 1\le l_1 < \ldots < l_{m_{j,k}} \le \ell_{j,k} \right\}, \nonumber\\
{\cal K}'_{\un{m}}= L\left\{ \bra{1}\prod_{j,k}\prod_{\eta=1}^{m_{j,k}} \bb_{j,k,l_\eta};\; 1\le l_1 < \ldots < l_{m_{j,k}} \le \ell_{j,k}\right\}. \label{eq:spaces}
\end{eqnarray}
(iii) However, the dimension of the eigenspace (the number of proper eigenvectors corresponding to $\lambda_\un{m}$) is smaller than $\dim {\cal K}_\un{m}$ in the nontrivial case
when at least one $\ell_{j,k} > 1$. The size of the largest Jordan block corresponding to $\lambda_\un{m}$ is
\begin{equation}
1+\sum_{j,k} (\ell_{j,k}-m_{j,k})m_{j,k} .
\label{eq:maxblock}
\end{equation}
\end{theorem}
\begin{proof}
We start by showing (ii). The fact that direct sums (\ref{eq:direct}) span the entire space follows from dimension counting $\sum_{\un{m}}\dim{\cal K}_{\un{m}}=\sum_{\un{m}}  \prod_{j,k}\left({\ell_{j,k}\atop m_{j,k}}\right)=\prod_{j,k}2^{\ell_{j,k}} = 2^{2n}$ and non-singularity of the complex orthogonal transformation to NMM maps (\ref{eq:bfroma}).
The invariance conditions (\ref{eq:inv}) are equivalent to 
\begin{equation}
\bra{y}\LL\ket{x} = 0, \quad \forall  \ket{x} \in {\cal K}_{\un{m}}, \bra{y} \in {\cal K}'_{\un{m}'},\; {\rm such\; that}\; \un{m} \ne \un{m}'\,,
\label{eq:orth}
\end{equation}
which in turn follows from the normal form (\ref{eq:dLL}). More explicitly, one may write (\ref{eq:dLL}) as $\LL = \sum_{j,k} \LL_{j,k}$ where
\begin{eqnarray}
\LL_{j,k} &:=& -2(\beta_j\, \NN_{j,k} + \MM_{j,k}),\\ 
\NN_{j,k} &=& \sum_{l=1}^{\ell_{j,k}}  \bb'_{j,k,l} \bb_{j,k,l},\quad
\MM_{j,k} = \sum_{l=1}^{\ell_{j,k}-1} \bb'_{j,k,l+1}  \bb_{j,k,l} \nonumber
\end{eqnarray}
and show (\ref{eq:orth}) separately for each term
$\bra{y}\LL_{j,k}\ket{x}=0$, and considering only $\ket{x} = \prod_{\eta=1}^{m_{j,k}} \bb'_{j,k,l'_\eta}  \ness, 
\bra{y} = \bra{1}\prod_{\eta=1}^{m_{j,k}} \bb_{j,k,l_\eta}$ for some $\{l_\eta\},\{l'_\eta\}$, since the other NMM maps $\hat{b}_{j',k',l'}$ for $(j',k')\neq (j,k)$ entering $\bra{x}$ and $\ket{y}$ (\ref{eq:spaces})
all annihilate as they commute with $\LL_{j,k}$. The rest is just trivial CAR algebra.

To show (i) we write the Liouvillean as $\LL = \LL_0 + \MM$ where $\LL_0 := -2\sum_{j,k} \beta_j\,\NN_{j,k}$ and $\MM = \sum_{j,k} \MM_{j,k}$.
Since $\MM_{j,k}$ are nilpotent (lemma \ref{nilplemma}), and $[\MM_{j,k},\MM_{j',k'}]=0$, $\MM$ is nilpotent as well.
 Clearly $[\LL_0,\MM]=0$ as all the terms commute $[\NN_{j,k},\MM_{j',k'}]=0$, so $\LL$ and $\LL_0$ should have identical spectra.
 But the eigenspaces of $\LL_0$ are ${\cal K}_{\un{m}}$, $\LL_0 {\cal K}_{\un{m}} = \lambda_{\un{m}} {\cal K}_{\un{m}}$, so $\{ \lambda_\un{m}\}$ is also the spectrum of $\LL$.
 
 As for (iii) we observe that the invariant space ${\cal K}_\un{m}$ admits a tensor product decomposition
 \begin{equation}
 {\cal K}_{\un m} \sim \bigotimes_{j,k} {\cal K}_{j,k,m_{j,k}}
 \end{equation}
where ${\cal K}_{j,k,m}$ is a $\left( \ell_{j,k}\atop m\right)$ dimensional space spanned by all $\prod_{\eta=1}^m \bb'_{j,k,l_\eta}\ness$, with  $1\le l_1 <\ldots< l_m \le \ell_{j,k}$.
Let us consider the restriction of the map on the invariant space 
$\LL|_{{\cal K}_{\un{m}}} = \lambda_{\un{m}} \one + \MM_{\un{m}}$ where the nilpotent $\MM_{\un{m}}$ acts independently on tensor factors
\begin{equation}
\MM_{\un{m}} = \sum_{j,k} \left( \bigotimes_{j',k'}^{(j',k') \neq (j,k)} \one_{{\cal K}_{j',k',m_{j',k'}}} \right) \otimes \MM^{(\ell_{j,k})}_{m_{j,k}}.
\end{equation} 
$\MM^{(\ell_{j,k})}_{m_{j,k}}$ is exactly the abstract nilpotent many-body map studied in subsection \ref{mbnilp} where $\ell = \ell_{j,k}$.
Noting the lemma \ref{nilplemma} each factor map has a nontrivial Jordan canonical form with a maximal block of size $(\ell_{j,k}-m_{j,k})m_{j,k}+1$.
Then we use lemma \ref{2nilp} recursively, noting at each step always only the largest Jordan block, to arrive at (\ref{eq:maxblock}).
\qed
\end{proof}
Writing down the complete Jordan canonical form of Liouvillean $\LL$ seems at this point a very difficult combinatorial problem even assuming the conjecture \ref{conjecture} were true.
We stress however an interesting non-trivial aspect of the result that we have just proven. Namely, in spite of the fact that the dimension of the invariant subspace pertaining to
some possibly highly degenerate rapidity $\beta_j$ is {\em exponentially large} in the (algebraic) multiplicity of the rapidity, the largest Jordan block is {\em only polynomial} (at most quadratic) in the multiplicity. This may have interesting physical consequences in studying Liouvillean time evolution (relaxation) where all terms of the form $t^{p-1} \exp( \lambda_{\un{m}}t)$ can appear
with integer power $p$ not larger than (\ref{eq:maxblock}).
 
\subsection{Uniqueness of NESS}

\begin{theorem} 
\label{uniqueness}
Uniqueness of NESS. $\ness$ is a unique stationary state of open quantum dynamics (\ref{eq:liouville}) if and only if all eigenvalues $\beta_j$ of $\mm{X}$ (\ref{eq:X}) lie away from the imaginary line $\re \beta_j > 0$.
If this is not the case, then:
\begin{enumerate}
\item
For each zero rapidity $\beta_j=0$, 
\begin{equation}
\ket{{\rm NESS};j,k} :=  \bb'_{j,k,1} \ness
\end{equation}
we also have the stationarity $\LL \ket{{\rm NESS};j,k} = 0$.
\item
For each imaginary rapidity $\beta_j=\ii b$, $b \in \RR\setminus\{0\}$, we have a corresponding negative rapidity $\beta_{j'}=-\ii b$, and
\begin{equation}
\ket{{\rm NESS};j,j',k,k'} := \bb'_{j,k,1} \bb'_{j',k',1} \ness,
\end{equation} which satisfies stationarity   $\LL \ket{{\rm NESS};j,j',k,k'} = 0$.
\end{enumerate}
\end{theorem}
\begin{proof}
Uniqueness of $\ness$ in case of strict positivity of the spectral gap of $\LL$, $\Delta = 2\min_j \re\beta_j$ (theorem \ref{spctheorem}) is obvious. In the other cases, (i), the stationarity $\LL \ket{{\rm NESS};j,k} = 0$ follows from the facts that the term $j,k$ is absent in the
normal form of the Liouvillean (\ref{eq:dLL}), since $\beta_{j,k} = 0$ and $\ell_{j,k}=1$ (lemma \ref{lemma:stability}), whereas
$\bb'_{j,k,1}$ anticommutes with all the other terms,. So the annihilation maps $\un{\bb}$ can be commuted to the right and then (\ref{eq:vacua}) is used.
Similarly, in case (ii), applying $\LL$ (\ref{eq:dLL}) to  $\ket{{\rm NESS};j,j',k,k'}$ and commuting annihilation maps $\un{\bb}$ to
the right results in two more terms w.r.t. case (i), which however cancel each other since $\beta_j+\beta_{j'}=0$. \qed
\end{proof}
Note that (\ref{eq:V}) implies
\begin{equation}
\un{\bb}' = \mm{P}^{-1} \un{\cc}^\dagger, \quad \un{\cc}^\dagger = \frac{1}{\sqrt{2}}(\un{\aaa}_1 + \ii \un{\aaa}_2)
\end{equation}
where $\un{\aaa}_1,\un{\aaa}_2$ represent, respectively, the first and the last $2n$ components of a $4n$ vector $\un{\aaa}$.
Therefore, since $\un{\cc}^\dagger_j$ can be interpreted as {\em creation} maps of {\em Hermitian} fermionic Majorana operators $w_j$ \cite{njp}, we have the following obvious statements: in case (i), the row $(j,k,1)$ of $\mm{P}^{-1}$ which is the {\em left} eigenvector of a real matrix $\mm{X}$ of eigenvalue $0$ can be always chosen {\em real}, and hence $\ket{{\rm NESS};j,k}$ represents a {\em Hermitian} operator
, and similarly in case (ii), the rows $(j,k,1)$ and $(j',k,1)$ of $\mm{P}^{-1}$ can be chosen to be mutually complex conjugate (as they correspond to left eigenvectors w.r.t. a pair of mutually conjugated eigenvalues) and hence $\ket{{\rm NESS};j,j',k,k}$ represents a {\em Hermitian} operator.
In other words, $\bb'_{j,k,1}$ (i) and $\bb'_{j,k,1}\bb'_{j',k,1}$ (ii) are, under the above conditions on the eigenvectors of $\mm{P}$, 
Hermiticity preserving maps. Furthermore, using similar arguments one can see that the {\em non-diagonal} maps
$\bb'_{j,k,1}\bb'_{j',k',1} + \bb'_{j,k',1}\bb'_{j',k,1}$ and $\ii(\bb'_{j,k,1}\bb'_{j',k',1} - \bb'_{j,k',1}\bb'_{j',k,1})$ are Hermiticity preserving, so
\begin{eqnarray}
\ket{{\rm NESS};j,j',k,k',+}:=\ket{{\rm NESS};j,j',k,k'} &+& \ket{{\rm NESS};j,j',k',k} \nonumber \\ 
\ket{{\rm NESS};j,j',k,k',-} :=\ii\ket{{\rm NESS};j,j',k,k'} &-& \ii\ket{{\rm NESS};j,j',k',k}
\end{eqnarray}
correspond to Hermitian operators.

Note as well that all these elements of ${\cal K}$ correspond to {\em trace-zero} operators as we have trivially
$\tr \rho_{{\rm NESS};j,k} = \braket{1}{{\rm NESS};j,k} = \bra{1}\bb'_{j,k,1}\ness = 0$,
$\tr \rho_{{\rm NESS};j,j',k,k'} = \braket{1}{{\rm NESS};j,j',k,k'} = \bra{1}\bb'_{j,k,1}\bb'_{j',k',1}\ness = 0$, due to (\ref{eq:vacua}).
Thus, in general a {\em convex set} of {\em real} coefficients should exist $\{ \alpha_{j,k},\alpha_{j,j',k,k',\pm}\in\RR\}$, such that
\begin{equation}
\ness + \sum_{j,k}^{\beta_j = 0} \alpha_{j,k} \ket{{\rm NESS};j,k} + \!\!\!\!\!\sum_{j,j',k,k',\tau=\pm}^{\beta_j+\beta_{j'}=0}\!\!\!\!\!\!\!\alpha_{j,j',k,k',\tau} \ket{{\rm NESS};j,j',k,k',\tau}
\end{equation}
represent valid positive trace-one density operators. 

Geometric characterization of this set of degenerate non-equilibrium steady states poses an interesting open problem.

\subsection{Physical observables in NESS}

\begin{theorem}
\label{observable}
Suppose the spectrum of $\mm{X}$ (\ref{eq:X}) is strictly away from the imaginary axis $\re \beta_j > 0$. Then
the complete set of quadratic correlators in NESS is given by the matrix $\mm{Z}$, a (unique) solution of the Lyapunov equation (\ref{eq:Lyap}),
\begin{equation}
\tr w_{j} w_{k} \rho_{\rm NESS} = \delta_{j,k} + 4 \ii Z_{j,k}.
\label{eq:corr}
\end{equation}
\end{theorem}
\begin{proof}
Inverting (\ref{eq:bfroma})
\begin{equation}
\un{\aaa} = \mm{V}^{-1}(\un{\bb}\oplus \un{\bb}')=\mm{V}^T(\un{\bb}'\oplus \un{\bb}),
\end{equation} 
we express the first first $2n$ components of the canonical Majorana maps explicitely
\begin{equation}
\un{\aaa}_1 = \frac{1}{\sqrt{2}}\left\{\mm{P}^{-T}\un{\bb} + (\one_{2n}+4\ii \mm{Z})\mm{P}\un{\bb}'\right\}.
\end{equation}
Then we note that, since $\ket{w_j w_k} = 2 \aaa_{1,j} \aaa_{1,k}\ket{1}$ and 
$\bra{w_j w_k} = 2 \bra{1} \aaa_{1,j} \aaa_{1,k}$, we can express the correlation matrix 
$C_{j,k} :=\tr w_{j} w_{k} \rho_{\rm NESS}$
as
\begin{eqnarray}
\mm{C} &=& 2\bra{1} \un{\aaa}_1\otimes \un{\aaa}_1\ness  \nonumber \\ 
&=& \bra{1}\left\{\mm{P}^{-T}\un{\bb} + (\one_{2n}+4\ii \mm{Z})\mm{P}\un{\bb}'\right\}\otimes\left\{\mm{P}^{-T}\un{\bb} + (\one_{2n}+4\ii \mm{Z})\mm{P}\un{\bb}'\right\}\ness \nonumber \\
&=& \mm{P}^{-T} \bra{1}\un{\bb}\otimes \un{\bb}'\ness[(\one_{2n}+4\ii\mm{Z})\mm{P}]^T = \one_{2n}+4\ii\mm{Z}
\end{eqnarray}
where we used (\ref{eq:vacua}), and $\bra{1}\un{\bb}\otimes \un{\bb}'\ness = \one_{2n}$ due to almost-CAR (\ref{eq:almostCAR}).
\qed
\end{proof}

\section{Examples}

We conclude this note by showing in two physical examples that the situations of non-diagonalizable Liouvilleans and subspaces of degenerate NESS are not as pathological or difficult to encounter as it might seem from the first sight.

\subsection{A single qubit with non-diagonalizable Liouvillean dynamics}

Take a single fermion (or qubit), described in terms of a pair of Hermitian anticomuting 
operators $w_1=c+c^\dagger$ and $w_2=\ii(c-c^\dagger)$.
Take Hamiltonian $H = \ii h w_1 w_2$ and a single Lindblad operator
$L = \sqrt{\Gamma} (w_1 + e^{\ii\theta} w_2)$
Then, the key matrices read
\begin{equation}
\mm{X} = 2 \pmatrix{ \Gamma & \Gamma \cos\theta +h \cr \Gamma \cos\theta - h & \Gamma},\;\;
\mm{M}_{\rm i} = \pmatrix{ 0 & -\Gamma\sin\theta \cr \Gamma \sin\theta & 0}.
\end{equation} 
For $h=\Gamma\cos\theta$ the matrix $\mm{X}$ becomes nondiagonalizable.

\subsection{A pair of spins with degenerate NESS}
Consider a pair  ($n=2$) of qubits, or spins $1/2$, described by Pauli operators $\un{\sigma}_{1,2}$,
coupled via Ising interaction
$H = \frac{1}{2}J \sigma^1_1 \sigma^1_2$
and attached to a Lindblad reservoir only through one of the qubits $L_1 = \Gamma_1 \sigma^+_1, L_2 = \Gamma_2 \sigma^-_1$.
This model is fermionized in terms of Majorana operators $w_{1,2,3,4}$ \cite{njp} using Jordan-Wigner transformation $H=-\frac{\ii}{2}  J w_2 w_3$ and the key matrices read 
\begin{equation}
\mm{X} = \pmatrix{ \Gamma_+ & 0 & 0 & 0 \cr 0 & \Gamma_+ & -J & 0 \cr 0 & J & 0 & 0 \cr 0 & 0 & 0 & 0},\;\;
\mm{M}_{\rm i} = \frac{\Gamma_-}{4}\pmatrix{ 0 & 1 & 0 & 0 \cr -1 & 0 & 0 & 0\cr 0 & 0 & 0 & 0\cr 0 & 0 & 0 & 0}
\end{equation}
where $\Gamma_{\pm} = \Gamma_2 \pm \Gamma_1$.
The rapidity spectrum is 
\begin{equation}
\beta_1 = \Gamma_+, \quad \beta_{2,3} = \frac{\Gamma_+}{2} \pm \sqrt{\left(\frac{\Gamma_+}{2}\right)^2 - J^2}, \quad \beta_4=0.
\end{equation}
The Lyapunov equation (\ref{eq:Lyap}) for the steady state 2-point correlator (\ref{eq:corr}) has a unique solution, in spite of zero eigenvalue (since it has multiplicity $1$, see the proof of lemma \ref{lemma:similar}), namely
\begin{equation}
\mm{Z} = \frac{2\Gamma_-}{2\Gamma_+^2+ J^2} 
\pmatrix{ 0 & \Gamma_+ & J & 0  \cr 
- \Gamma_+ & 0 & 0 &  0 \cr
-J & 0 & 0 & 0 \cr
0 & 0 & 0 & 0}.
\end{equation}
However, the non-uniqueness of $\ness$ is manifested in non-vanishing and non-unique $1-$point functions, namely since the non-unique NESS is
parametrized as $\ket{{\rm NESS};\alpha} = \ness + \alpha \ket{{\rm NESS};4,1} $, where $\ket{{\rm NESS};4,1} = \hat{b}'_4 \ness = \hat{c}^\dagger_4 \ness$, we have
$\ave{w_j} = \delta_{j,4} \alpha$.  Using the (generalized) Wick theorem, we see that even order monomials are insensitive to non-uniqueness parameter
$\alpha$, whereas odd order monomials depend on $\alpha$.
The range of admissible values of $\alpha$ is obtained for example by explicit construction of the spectrum of the NESS density matrix
$\rho_{\rm NESS} = 4\sum_{\gamma_j\in\{0,1\}} \bra{1}\hat{a}_{1,1}^{\gamma_1}\hat{a}_{1,2}^{\gamma_2}\hat{a}_{1,3}^{\gamma_3}\hat{a}_{1,4}^{\gamma_4}\ket{{\rm NESS};\alpha}
w^{\gamma_1}_1 w^{\gamma_2}_2 w^{\gamma_3}_3 w^{\gamma_4}_4$
and require it to lie in the interval $[0,1]$.

Richer structure of degenerate NESS is obtained by adding further spins $1/2$ to such a one-sided open Ising chain, namely each new spin coupled to the rest of the chain
via Ising interaction adds a new pair of completely imaginary structure matrix eigenvalues $\pm i b$ (plus a pair of fully complex eigenvalues).
 
\section*{Acknowledgements}

The author is grateful to Iztok Pi\v zorn for providing some valuable information on the structure of the spectrum and eigenvectors of the Liouvillean structure matrix which motivated this work, and to Peter \v Semrl for extensive discussions and providing detailed arguments for the conjecture \ref{conjecture} underlying the Jordan canonical form structure of the 
many-body nilpotent map. The work is supported by the Programme P1-0044 and the Grant J1-2208 of the Slovenian Research Agency.
Finally, hospitality of University of Potsdam 
in the final stage of preparing this publication is warmly acknowledged.

 \end{document}